\documentclass[english]{article}
\usepackage[latin1]{inputenc}
\usepackage{geometry}
\usepackage{mathrsfs}
\usepackage{amsmath}
\usepackage{amssymb}
\usepackage{amsthm}

\usepackage{amsfonts}
\usepackage{epsfig}
\usepackage{graphicx}
\usepackage{algorithm}
\usepackage[noend]{algpseudocode}
\usepackage{babel}
\usepackage{fullpage}
\usepackage[colorlinks,citecolor=blue,urlcolor=blue]{hyperref}

\renewcommand{\phi}{\varphi}
\renewcommand{\epsilon}{\varepsilon}
\newcommand{\FF}{\mathcal{F}}

\newcommand{\RR}{\mathbb{R}}

\newcommand{\X}{\ensuremath{\mathcal{X}}}

\newcommand{\assetvalue}{\mathbf R}
\newcommand{\logassetvalue}{\mathbf S}
\newcommand{\logassetvaluesmall}{\mathbf s}
\newcommand{\nparticles}{M}
\newcommand{\particleindex}{m}
\newcommand{\ntimepoints}{p}

\newcommand{\EE}{\ensuremath{\operatorname{E}}}
\newcommand{\E}[1]{\EE\left[ #1  \right]}
\newcommand{\CE}[2]{\EE\left[\left. #1 \, \right| #2 \right]}
\newcommand{\CCE}[3]{\EE_{#1}\left[\left. #2 \, \right| #3 \right]}

\newcommand{\BK}[1]{ {\left( #1 \right)} }          
\newcommand{\sqBK}[1]{ {\left[ #1 \right]} }       

\newcommand{\Normal}[1]{\ensuremath{\operatorname{N}}\BK{#1}}
\newcommand{\dNormal}[1]{\ensuremath{\phi}\BK{#1}}

\newtheorem{rem}{Remark}[section]
\newtheorem{prop}{Proposition}[section]

\newcounter{hypA}

\date{}

\begin{document}

\begin{center}

{\Large \textbf{Some Contributions to Sequential Monte Carlo Methods for Option Pricing}}

\vspace{0.5cm}
\author{blabla}
BY DEBORSHEE SEN \footnote{Corresponding author.}, AJAY JASRA \& YAN ZHOU

{\footnotesize Department of Statistics \& Applied Probability,
National University of Singapore, Singapore, 117546, SG.}
{\footnotesize Email:\,}\texttt{\emph{\footnotesize deborshee.sen@u.nus.edu; staja@nus.edu.sg; stazhou@nus.edu.sg}}\\
\end{center}

\begin{abstract}
Pricing options is an important problem in financial engineering.
In many scenarios of practical interest, financial option prices associated to an underlying asset reduces to computing an expectation
w.r.t.~a diffusion process. In general, these expectations cannot be calculated analytically, and one way to approximate these quantities is via the Monte Carlo method; Monte Carlo methods have been used to price options since at least the 1970's. It has been seen in \cite{valuation, SMC_option_jasra} that  Sequential Monte Carlo (SMC) methods are a natural tool to apply in this context and can vastly improve over standard Monte Carlo. In this article, in a similar spirit to \cite{valuation, SMC_option_jasra} we show that one can achieve significant gains by using SMC methods by constructing a sequence of artificial target densities over time. In particular, we approximate the optimal importance sampling distribution in the SMC algorithm by using a sequence of weighting functions. This is demonstrated on two examples, barrier options and target accrual redemption notes (TARN's).  We also provide a proof of unbiasedness of our SMC estimate.\\
\textbf{Key words:}  Diffusions; Sequential Monte Carlo; Option Pricing \\
\textbf{AMS Subject Classification:} Primary 91G60; Secondary 65C05.  
\end{abstract}

\subsubsection*{Acknowledgements}
AJ was supported by a Singapore Ministry of Education Academic Research Fund Tier 1 grant (R-155-000-156-112) and is affiliated with the RMI and CQF at NUS. YZ was supported by a Singapore Ministry of Education Academic Research Fund Tier 2 grant (R-155-000-143-112).

\section{Introduction}
A basic (or \textit{vanilla}) option is a financial product which provides the holder of the option with the right to buy or sell a specified quantity of an underlying asset at a fixed price on or before the expiration date of the option. There are many more complex options (called \textit{exotic} options) in use today; these are often of more practical interest and are harder to deal with.  
In many scenarios of practical interest and as we shall assume in this article, the value of the underlying asset can be described by a diffusion process; in a complete market, the value of the option can be expressed as the expectation under the risk neutral probability of a functional of the paths of the underlying diffusion process.
In general these expectations cannot be calculated analytically. The Monte Carlo (MC) method is a standard approach used to approximate these quantities and it has been extensively used in option pricing since \cite{boyle}. Subsequently, a wide variety of Monte Carlo approaches have been applied (\cite{glasserman} provides a thorough introduction). \\

The importance of Monte Carlo for option pricing against other numerical approaches is its ability to deal with high-dimensional integrals. This is either in the time parameter of the option (path-dependent options) or in the dimension of the underlying (basket of options), and more generally in both. However, it has been noted in the option pricing literature that standard Monte Carlo estimates can suffer from high variability. 
It has been seen in  \cite{valuation,SMC_option_jasra, jasra_doucet} that, in many situations of practical interest, Sequential Monte Carlo (SMC) approaches can vastly improve over more standard Monte Carlo techniques. \\

Sequential Monte Carlo methods are a general class of methods to sample from a sequence of distributions of increasing dimensions which  have extensively been used in engineering, statistics, physics and other domains. \cite{tutorial_15_years_later} provides an introduction and shows how essentially all methods for particle filtering can be interpreted as some special instances of a generic SMC algorithm. SMC methods make use of  a sequence of proposal densities to sequentially approximate the targets via a collection of $\nparticles$ samples, termed particles. In most scenarios it is not possible to use the distribution of interest as a proposal. Therefore, one must correct for the discrepancy between proposal and target  via importance weights. In the majority of cases of practical interest, the variance of these importance weights increases with algorithmic time. This can, to some extent, be dealt with a  resampling procedure consisted of sampling with replacement from the current weighted samples and resetting them to $1/\nparticles$ (adaptive resampling). The variability of the weights is often measured by the effective sample size (ESS).  Several convergence results, as $\nparticles$ grows, have been proved \cite{hock_peng_martingale,del_moral_book,delmoral1,douc}. SMC methods have also recently been proven to be stable in certain high-dimensional contexts \cite{beskos}. \\

The main contributions of this paper are as follows.
We develop the formal framework of \emph{weighting functions}; this technique has already been used, implicitly or explicitly, in \cite{del2005genealogical,valuation,SMC_highly_informative,SMC_option_jasra}. Exploiting this framework, we develop tailored methods for pricing of barrier options in high dimensional settings. It is also applied to the pricing of Target Accrual Redemption Note (TARN) which are another widely traded kind of path dependent options that are notoriously difficult to accurately value. On the theoretical side, we provide with a proof of the unbiasedness of the SMC estimates when an adaptive resampling scheme is used. \\

This paper is structured as follows. In Section \ref{sec:options} we provide background details on option pricing. In Section \ref{sec:smc} we give a basic summary of SMC methods. In Section \ref{sec:weighting} we give the weighting functions framework and its application in the context of our option pricing problems. In Section
\ref{sec:numerics} our methods are illustrated numerically. 
The appendix gives the proof of unbiasedness of our SMC estimate in the adaptive resampling case. \\

In the remainder of this article, we use the notation $\RR^{d}$ to denote the $d$-dimensional Euclidean space and $\RR_{+}^{d} \equiv (0,\infty)^{\otimes d}$. A normal distribution with mean $\mu$ and variance $\sigma^{2}$ is denoted by $\Normal{\mu, \sigma^2}$ and its density at $x$ is denoted by by $\dNormal{x; \mu, \sigma^2}$. $I_{d}$ denotes the $d$-dimensional identity matrix. $\EE$ denotes expectation.

\section{Option pricing}
\label{sec:options}

Options come in two basic kinds - call and put. Call options give the right to buy and put options give the right to sell. In this context, there are two main kinds of options - American and European. American options can be exercised at any time prior to expiration whereas European options can be exercised only at expiration. We focus on European options in this paper. European call/put options are known as vanilla options since they are relatively simple in structure. An exotic option is an option which has features making it more complex than commonly traded vanilla options. Path-dependent options are an example, in which case the payoff depends on the value of the underlying at some (or all) time points prior to the expiration date. We consider two kinds of path-dependent options in this paper, namely barrier options and TARN's, which we shall describe shortly.

\subsection{Model}

Consider a collection of $d$ underlying assets; this is also known as a basket. We denote by $\assetvalue_{t} \in \RR^{d}$ the value of the assets in the basket. $\assetvalue_{t}$ is typically modelled by a diffusion process. One such process is a Black-Scholes model with a drift and a volatility
\begin{equation} \label{eq:pde}
d \assetvalue_{t} = \boldsymbol \mu(t,\assetvalue_{t})  dt \assetvalue_{t} + \boldsymbol\sigma(t,\assetvalue_{t}) \assetvalue_{t} d\mathbf{W}_{t},
\end{equation}
where $ \boldsymbol\mu: \RR_{+} \times \RR^{d} \rightarrow \RR^{d} $ is the drift function, $ \boldsymbol\sigma: \RR_{+} \times \RR^{d} \rightarrow \RR^{d} $ is the volatility and $\mathbf{W}_{t}$ denotes a Brownian motion in $\RR^{d}$ with mean $\mathbf{0}$ and covariance matrix $\Sigma$. It is reasonable to assume that $\mathbf{W}_{t}$ is normalized, that is, $\Sigma_{i,i} = 1$ for all $i$; this assumption is valid because the scale factor can be included in the volatility term. There is an interest rate $r$, which can depend on time as well. \\

In general, it is hard to analytically work with \eqref{eq:pde} except in simple scenarios. This has lead to several discretization methods being available in literature (\cite{glasserman}) with varying levels of accuracy and complexity. One of the most widely used discretization methods is the Euler-Maruyama discretization and we work with it in this paper; however other discretization schemes could also be used which could lead to a lower bias. Consider discretized time points $0 = t_{0} < t_{1} < \cdots < t_{N} = T$. By letting $\logassetvalue$ denote the logarithm $\assetvalue$ and writing $\logassetvalue_{n}$ in place of $\logassetvalue_{t_{n}}$, the Euler-Maruyama discretization of \eqref{eq:pde} is
\begin{equation} \label{eq:euler_maruyama}
\logassetvalue_{n+1} = \logassetvalue_{n} + \left ( \boldsymbol\mu(n,\logassetvalue_{n}) - \frac{1}{2} \boldsymbol\sigma^{2}(n,\logassetvalue_{n}) \right ) \delta t_{n} + \boldsymbol\sigma(n,\logassetvalue_{n}) \sqrt{\delta t_{n}} \mathbf{Z}_{n}
\end{equation} 
for $n = 1, \ldots, N$, where $\mathbf{Z}_{n} \sim N_{d}(\mathbf{0},\Sigma)$, $\mathbf{0}$ denotes the origin in $\RR^{d}$, and $\delta t_{n} = t_{n+1} - t_{n}$\footnote{We have a slight abuse of notation in the above, wherein we have used $\boldsymbol\mu(n,\logassetvalue_{n})$ and $\boldsymbol\sigma(n,\logassetvalue_{n})$ to denote $\boldsymbol\mu(t_{n},\assetvalue_{t_{n}})$ and $\boldsymbol\sigma(t_{n},\assetvalue_{t_{n}})$ respectively.}. \\

We do not focus on the level of discretization here. Given a particular discretization, we apply our methods to it. The methods developed here could however be used in a multilevel setup (as in \cite{giles}); we do not explore this further in this paper. \\

We assume the drift $\boldsymbol\mu = \mathbf{0}$ in order to keep things simple. \footnote{If $\boldsymbol\mu$ is a constant other than $\mathbf{0}$, then it is trivial to extend the methods we propose. If it is a function of the asset value, we could do things similar to what we do in the local volatility model considered later.} We work with two cases, one where the volatility $\boldsymbol\sigma$ is a constant, and another where it depends on the price of the underlying. We now describe two kinds of path-dependent options and we shall later demonstrate our methods on these.

\subsection{Barrier options}

A barrier option is an exotic derivative, typically an option, on the underlying asset(s) whose value(s) on reaching pre-set barrier levels either springs the option into existence or extinguishes an already existing option. Barrier options exist for baskets as well and the barrier conditions may in general be defined as a function of the underlying assets. For example, the function of the underlying assets could be the mean or could be the maximum of their values. There are two kinds of barrier options:
\begin{itemize}
\item when the option springs into existence if a function of the underlying asset values breaches prespecified barriers, it is referred to as being `knocked-in', and
\item when the option is extinguished if a function of the underlying asset values breaches prespecified barriers, it is referred to as being `knocked-out'. 
\end{itemize}  

We consider knocked-out options. These are options that are `alive' as long as the values of the underlying satisfy barrier conditions at some (or all) time points prior to the expiration. If the barrier condition is breached, the option `dies' leading to a zero payoff. If the option is still `alive' at the expiration date, then it gives a payoff that is akin to either a call or a put option (depending on the option type).  \\

We consider a basket of $d$ underlying assets and a barrier option based on these. Barrier options are hard to price using standard MC in the sense that most (if not all) of the particles lead to a zero payoff and this contributes to a high variance of the final estimate. Because of the sequential nature of the evolution of the asset values over time, this is a natural example of a setting where SMC methods can be applied; indeed \cite{valuation} talks about how SMC methods can be used in this context. We extend their method and show that one can obtain significant gains by choosing the $h_{n}$'s in Section 5 of their paper even by heuristic methods.

\subsection{Target Accrual Redemption Note}

A Target Accumulation Redemption Note (TARN) provides a capped sum of payments over a period with the possibility of early termination (knockout) determined by target levels imposed on the accumulated amounts. A certain amount of payment is made on a series of cash flow dates (referred to as fixing dates) until the target level is breached. The payoff function of a TARN is path dependent in that the payment on a fixing date depends on the spot value of the asset as well as on the accumulated payment amount up to the fixing date. Typically, commercial software solutions for pricing TARN's are based on the MC method. \\

There are different versions of TARN products used in FX trading. For simplicity, we consider here a specific form of TARN's. Consider a sequence of fixing dates $0 < T_{1} < T_{2} < \cdots < T_{\ntimepoints} = T$ and a function $f: \RR^{d} \rightarrow \RR$. The function $f$ is decomposed into its positive and negative parts as $f = f^{+} - f^{-}$. Gain and loss processes are defined as follows:
$$ G_{k} = \sum_{i=1}^{k} f^{+} (\assetvalue_{T_{i}}) \hspace{0.1in} \text{and} \hspace{0.1in} L_{k} = \sum_{i=1}^{k} f^{-} (\assetvalue_{T_{i}}), $$
these are the amounts of positive and negative cashflows respectively. There are two cashflow cutoffs $\Gamma_{G}$ and $ \Gamma_{L}$. Stopping times $\tau^{(L)}$ and $\tau^{(G)}$ are defined as 
$$ \tau^{(L)} := \inf \{k: L_{k} \geq \Gamma_{L} \} \hspace{0.1in} \text{and} \hspace{0.1in} \tau^{(G)} := \inf \{k: G_{k} \geq \Gamma_{G} \}, $$ 
these are the first times when the positive and negative cash flows cross their cut-offs respectively. The overall stopping time $\tau$ is defined as 
$$ \tau := \inf \{ \tau^{(L)}, \tau^{(G)}, \ntimepoints \}, $$
which is the first time either the positive or the negative cash flows cross a cut-off. The price of the TARN is the expected value of the overall cash flow, $\mathbb{E} [ \sum_{i=1}^{\tau} f(\assetvalue_{T_{i}}) ]$. Here we have assumed that the interest rate is $0$, if it wasn't then the expectation would be a weighted sum with the weights corresponding to discounting factors. The reason we assume the interest rate is $0$ will be explained in Section 4.3. The main difficulty in applying standard MC methods in this scenario arises from the fact that the function $f$ may be discontinuous. SMC methods can be used instead, which we shall show later.

\section{Sequential Monte Carlo methods}\label{sec:smc}

SMC methods are a general class of MC methods that sample sequentially from a sequence of target probability densities $\{ \pi_{n}(x_{1:n} ) \}_{n \geq 1} $ of increasing dimension, where each distribution $\pi_{n} (x_{1:n}) $ is defined on the product space $\X^{n}$. Writing 

\begin{equation} \nonumber
\pi_{n}(x_{1:n}) = \frac{ \gamma_{n}(x_{1:n})}{Z_{n}},
\end{equation}
it is required only that $\gamma_{n}:\X^{n} \rightarrow \RR $ be known pointwise. In particular, the normalizing constants $Z_{n} = \int_{\X^{n}} \gamma_{n}(x_{1:n}) dx_{1:n}$ may be unknown and an estimate is obtained by the SMC method. SMC provides an approximation of $\pi_{1}(x_{1})$ and an estimate of $Z_{1}$ at time 1, then an estimate of $\pi_{2}(x_{1:2})$ and an estimate of $Z_{2}$ at time 2, and so on. These methods work by propagating a collection of $\nparticles$ particles using a sequence of importance sampling and resampling steps. In what follows, superscript $^{(\particleindex)}$ shall denote the $\particleindex$-th particle. When we write an operation with superscript $^{(\particleindex)}$, we mean that it happens for all $\nparticles$ particles. A proposal density $q_{n}(x_{1:n})$ is selected and particles are proposed from this. The proposal density has the following structure:
\begin{equation} \nonumber
q_{n}(x_{1:n}) = q_{n-1}(x_{1:n-1}) q_{n}(x_{n}|x_{1:n-1}) = q_{1}(x_{1}) \prod_{k=1}^{n} q_{k}(x_{k}|x_{1:k-1}). 
\end{equation}
After proposing particles, associated \textit{unnormalized} weights are computed recursively using the decomposition
\begin{equation} \nonumber
w_{n}(x_{1:n}) = \frac{ \gamma_{n}(x_{1:n}) }{ q_{n}(x_{1:n} ) } = \frac{ \gamma_{n-1}(x_{1:n-1}) }{ q_{n-1}(x_{1:n-1}) } \frac{ \gamma_{n}(x_{1:n}) }{ \gamma_{n-1}(x_{1:n-1}) q_{n}(x_{n}|x_{1:n-1} ) }.
\end{equation}
These can be written in the form 
\begin{equation} \nonumber
w_{n} (x_{1:n}) = w_{n-1}(x_{1:n-1}) \alpha_{n} (x_{1:n}) = w_{1}(x_{1}) \prod_{k=1}^{n} \alpha_{k}(x_{1:k}),
\end{equation}
where the \textit{incremental weight} function $\alpha_{k}(x_{1:k})$ is given by 
$$ \alpha_{k}(x_{1:k}) = \frac{ \gamma_{k}(x_{1:k}) }{ \gamma_{k-1}(x_{1:k-1}) q_{k}(x_{k}|x_{1:k-1} ) }. $$
These are computed for each particle.
The weights $w_{n}(x_{1:n})$'s are unnormalized because they do not add up to 1. They are normalized by dividing by their sum and the normalized weights are denoted by $W_{n}$'s. \\

Once we obtain a collection of $\nparticles$ weighted particles, they are \textit{resampled} according to their weights $\left \{ W_{n}^{(\particleindex)} \right \}_{\particleindex=1}^{\nparticles}$. This is a crucial step which ensures that the system doesn't collapse to very few particles with very high weights. Resampling indices $\left \{ I_{n}^{(\particleindex)} \right \}_{\particleindex=1}^{\nparticles}$ are chosen such that $\mathbb{P} \left ( I_{n}^{(\particleindex)} = j \right ) = W_{n}^{(j)}$. The system is then updated by setting $ \overline{X}_{1:n}^{(\particleindex)} \leftarrow X_{1:n}^{\left ( I_{n}^{(\particleindex)} \right )}$ and $\overline{W}_{n}^{(\particleindex)} \leftarrow 1/\nparticles$ for $1 \leq \particleindex \leq \nparticles$. Resampling is expensive and in practice is done only if the variance of the weights is high. One such method is to use the Effective Sample Size (ESS), defined at time $n$ as 
$$ \text{ESS} = \frac{1}{ \sum_{\particleindex=1}^{\nparticles} \left ( W_{n}^{(\particleindex)} \right )^{2} } $$
and to perform resampling at time $n$ if the ESS falls below a certain threshold, usually taken to be $\nparticles/2$ or $\nparticles/3$. Algorithm \ref{algo:SMC} describes a general SMC algorithm.

\begin{algorithm}
\caption{A generic SMC algorithm}\label{algo:SMC}
\begin{algorithmic}[1]

\State Set initial weights $\overline{W}_{0}^{(\particleindex)} = 1/ \nparticles$ and initial estimate of normalizing constant $\widehat{C}_{0} = 1$.

\For {$n=1$}

\State Sample $X_{1}^{(\particleindex)} \sim q_{1}(x_{1})$,

\State compute unnormalized weights $w_{1} \left ( X_{1}^{(\particleindex)} \right )$,

\State update estimate of normalizing constant $ \widehat{C}_{1} = \widehat{C}_{0} \times \sum_{\particleindex=1}^{\nparticles} w_{1}\left ( X_{1}^{(\particleindex)} \right )$,

\State compute normalized weights $W_{1}^{(\particleindex)} \propto w_{1} \left ( X_{1}^{(\particleindex)} \right )$,

\If {resampling criterion is satisfied} 

\State resample $ \left \{ W_{1}^{(\particleindex)}, X_{1}^{(\particleindex)} \right \}_{\particleindex=1}^{\nparticles}$ to obtain $\nparticles$ equally-weighted particles $\left \{ 1/\nparticles, \overline{X}_{1}^{(\particleindex)} \right \}_{\particleindex=1}^{\nparticles}$ and set $ \left \{ \overline{W}_{1}^{(\particleindex)}, \overline{X}_{1}^{(\particleindex)} \right \}_{\particleindex=1}^{\nparticles} \leftarrow \left \{ 1/\nparticles, \overline{X}_{1}^{(\particleindex)} \right \}_{\particleindex=1}^{\nparticles}$,

\Else

\State set $ \left \{ \overline{W}_{1}^{(\particleindex)}, \overline{X}_{1}^{(\particleindex)} \right \}_{\particleindex=1}^{\nparticles} \leftarrow \left \{ W_{1}^{(\particleindex)}, X_{1}^{(\particleindex)} \right \}_{\particleindex=1}^{\nparticles}$.

\EndIf

\EndFor

\For {$n \geq 2$}

\State Sample $X_{n}^{(\particleindex)} \sim q_{n}\left ( x_{n}|\overline{X}_{1:n-1}^{(\particleindex)} \right )$ and set $ X_{1:n}^{(\particleindex)} \leftarrow \left ( \overline{X}_{1:n-1}^{(\particleindex)}, X_{n}^{(\particleindex)} \right )$,

\State compute incremental weights $ \alpha_{n} \left ( X_{1:n}^{(\particleindex)} \right )$ and unnormalized weights $ w_{n} \left ( X_{1:n}^{(\particleindex)} \right ) = \overline{W}_{n-1}^{(\particleindex - 1)} \alpha_{n} \left ( X_{1:n}^{(\particleindex)} \right )$,

\State update estimate of normalizing constant $ \widehat{C}_{n} = \widehat{C}_{n-1} \times \sum_{\particleindex=1}^{\nparticles} w_{n}\left ( X_{1:n}^{(\particleindex)} \right )$,

\State compute normalized weights $W_{n}^{(\particleindex)} \propto w_{n} \left ( X_{1:n}^{(\particleindex)} \right )$,

\If {resampling criterion is satisfied} 

\State resample $ \left \{ W_{n}^{(\particleindex)}, X_{1:n}^{(\particleindex)} \right \}_{\particleindex=1}^{\nparticles}$ to obtain $\nparticles$ equally-weighted particles $\left \{ 1/\nparticles, \overline{X}_{1:n}^{(\particleindex)} \right \}_{\particleindex=1}^{\nparticles}$ and set $ \left \{ \overline{W}_{n}^{(\particleindex)}, \overline{X}_{1:n}^{(\particleindex)} \right \}_{\particleindex=1}^{\nparticles} \leftarrow \left \{ 1/\nparticles, \overline{X}_{1:n}^{(\particleindex)} \right \}_{\particleindex=1}^{\nparticles}$,

\Else 

set $ \left \{ \overline{W}_{n}^{(\particleindex)}, \overline{X}_{1:n}^{(\particleindex)} \right \}_{\particleindex=1}^{\nparticles} \leftarrow \left \{ W_{n}^{(\particleindex)}, X_{1:n}^{(\particleindex)} \right \}_{\particleindex=1}^{\nparticles}$.

\EndIf

\EndFor

\end{algorithmic}
\end{algorithm}

We have two approximations of $\pi_{n}(x_{1:n})$:
\begin{eqnarray}
\widehat{\pi}_{n}(x_{1:n}) & = & \sum_{\particleindex=1}^{\nparticles} W_{n}^{(\particleindex)} \delta_{X_{1:n}^{(\particleindex)}}(x_{1:n}) \textrm{ and} \nonumber \\
\bar{\pi}_{n}(x_{1:n}) & = & \sum_{\particleindex=1}^{\nparticles} \overline{W}_{n}^{(\particleindex)} \delta_{\overline{X}_{1:n}^{(\particleindex)}} (x_{1:n}), \nonumber 
\end{eqnarray} 
which are equal if no resampling is used at time $n$. Here $\delta$ denotes the Dirac delta function. We can also estimate $Z_{n}/Z_{n-1}$ through 
$$ \widehat{ \frac{Z_{n}}{Z_{n-1}} } = \sum_{\particleindex=1}^{\nparticles} \overline{W}_{n-1}^{(\particleindex)} \alpha_{n} \left ( X_{1:n}^{(\particleindex)} \right ). $$ 

\section{Using weighting functions}\label{sec:weighting}

\subsection{In SMC}

The goal of using weighting functions is to create a sequence of intermediate target densities that try to approximate the optimal importance sampling density to guide particles towards regions of interest. Consider a discrete time stochastic process $ X_{1:N} $ on a general space $\X^{N}$. Suppose that we want to estimate the expectation of a function of this process $H(X_{1:N})$. The basic Monte Carlo method simulates $\nparticles$ independent realizations of the process $ \left \{ X_{1:N}^{(\particleindex)} \right \}_{\particleindex=1}^{\nparticles}$, where $X_{1:N}^{(\particleindex)}$ denotes the $\particleindex$-th realization. $\E{ H(X_{1:N} ) }$ is estimated by 
$$ \frac{1}{\nparticles} \sum_{\particleindex=1}^{\nparticles} H \left ( X_{1:N}^{(\particleindex)} \right ). $$ 
If $H$ is of the form where $H(X_{1:N})$ is zero over most of the space $\X^{N}$ and non-zero only on a small subset, most of the $H \left ( X_{1:N}^{(\particleindex)} \right )$'s would be zero. This subsequently leads to a high variance of the resulting estimate. \\

More specifically, since $\X^{N}$ is the sample space, we write 
$$ \mathcal{S} := \X^{N} = \mathcal{S}_{1} \cup \mathcal{S}_{2}, $$
where 
$$ \mathcal{S}_{1} = \left \{ x_{1:N} \in \X^{N} : H(x_{1:N}) = 0 \right \} $$ 
and 
$$ \mathcal{S}_{2} = \left \{ x_{1:N} \in \X^{N} : H(x_{1:N}) \neq 0 \right \}. $$ 
If $\mathbb{P} ( X_{1:N} \in \mathcal{S}_{1} )$ is much larger than $\mathbb{P} ( X_{1:N} \in \mathcal{S}_{2} ) $, simulating $M$ independent realizations will lead to most of them lying in $\mathcal{S}_{1}$. Our goal is to use SMC to simulate more particles from $\mathcal{S}_{2}$. In order to do that, we consider a sequence of positive potential functions  $G_{n}: \mathcal{X}^{n} \rightarrow \RR$, $n = 1, 2, \ldots, N$, such that $\prod_{n=1}^{N} G_{n}(X_{1:n}) = 1$ and write
$$ H(X_{1:N}) = \left [  \prod_{n=1}^{N} G_{n}(X_{1:n}) \right ] H(X_{1:N}). $$ \\
The goal is to choose the $G_{n}$'s such that they guide the particles towards being in $\mathcal{S}_{2}$ through the weighting and resampling steps of SMC. \\

When this is done on a path space, the resulting algorithm is sometimes known as tempering. \cite{SMC_samplers} considers this and shows how one can construct an artificial sequence of intermediate target densities on the path space which guide particles towards regions of interest. Doing it on the path space however makes it computationally expensive and we do not work on the path space.

\subsection{Barrier options} \label{sec:barrier_SMC}

We consider a discretely monitored barrier option monitored on a series of monitoring dates $T_{1} < \cdots < T_{\ntimepoints} = T$ and a sequence of lower and upper barriers $\mathbf{L}_{T_{1}}, \ldots, \mathbf{L}_{T_{\ntimepoints}} \in \RR^{d}$ and $\mathbf{U}_{T_{1}}, \ldots, \mathbf{U}_{T_{\ntimepoints}} \in \RR^{d}$ respectively. We suppose that the barrier conditions are that all the underlying asset values lie inside their respective barriers at the monitoring times. This is a simplistic assumption and makes it easier for us to demonstrate our methods; more complicated barrier conditions could also be used. \\

For ease of notation, we remove the $T_{i}$ from the barriers and replace it simply by $i$. Let $X_{n} = (\logassetvalue_{n-1},\logassetvalue_{n}) \in \RR^{2d}$ for $n \geq 1$ and let $H$ denote the payoff function at time $T$. The sequence of random variables $ X_{1:N} $ then forms a Markov Chain. The price of the barrier option is

$$ Q_{D} = \E{ H(\logassetvalue_{T}) \prod_{i=1}^{\ntimepoints} \mathbf{1} \left \{ \logassetvalue_{T_{i}} \in (\mathbf{L}_{i}, \mathbf{U}_{i}) \right \} } \footnote{We have assumed here that the interest rate is 0. If the interest rate was $r$, then there would be a factor of $e^{\int_{0}^{T} r(t) d t }$ multiplied with $Q_{D}$. This is a constant and affects the variance of the estimate only upto a (known) scale factor.}, $$
where 
$$ \mathbf{1} \left \{ \logassetvalue_{T_{i}} \in (\mathbf{L}_{i}, \mathbf{U}_{i}) \right \} = \prod_{j=1}^{d} \mathbf{1} \left \{ S_{T_{i},j} \in (L_{i,j}, U_{i,j}) \right \}. $$ 
$S_{T_{i},j}, L_{i,j}$ and $U_{i,j}$ denote the $j$-th components of $\logassetvalue_{T_{i}}, \mathbf{L}_{i}$ and $\mathbf{U}_{i}$ respectively. In this case, 
\begin{eqnarray}
\mathcal{S} & = & \RR^{dN}, \nonumber \\
\mathcal{S}_{2} & = & \left \{ \mathbf{s} \in \RR^{dN}: s_{T_{i},j} \in (L_{i,j}, U_{i,j}), \textrm{ } i = 1, 2, \ldots, \ntimepoints, \textrm{ } j = 1, 2, \ldots, d \right \}, \nonumber \\
\mathcal{S}_{1} & = & \mathcal{S} \setminus \mathcal{S}_{2}. \nonumber 
\end{eqnarray}

The authors in $\cite{valuation}$ introduce an SMC algorithm to estimate the price. The proposal density $q$ is chosen to be density with respect to the underlying discretization. For each time interval $(T_{i-1},T_{i})$ they simulate forward till $T_{i}$ and then resample particles that breach the barrier condition at time $T_{i}$ from among the particles that are still inside the barrier. Their algorithm is Algorithm \ref{algo:SMC_barrier}. It is commented that they do not use an adaptive version of resampling and instead always resample.

\begin{algorithm}
\caption{SMC For barrier options}\label{algo:SMC_barrier}
\begin{algorithmic}[1]

\State Set initial weights $\overline{W}_{0}^{(\particleindex)} = 1/ \nparticles$ and initial estimate of normalizing constant $\widehat{C}_{0} = 1$.

\For {$n=1$}

\State Sample $X_{1}^{(\particleindex)} \sim q_{1}(x_{1})$.

\EndFor

\For {$n \geq 2$}

\State Sample $X_{n}^{(\particleindex)} \sim q_{n}\left ( x_{n}|\overline{X}_{n-1}^{(\particleindex)} \right )$ and set $ X_{1:n}^{(\particleindex)} \leftarrow \left ( \overline{X}_{1:n-1}^{(\particleindex)}, X_{n}^{(\particleindex)} \right )$.

\If {$n \in \{T_{1}, \ldots, T_{\ntimepoints} \}$, say $n = T_{k}$}

\State Compute unnormalized weights $w_{k}^{(\particleindex)} = W_{k-1}^{(\particleindex)} \times \mathbf{1} \left \{ \logassetvalue_{T_{k}}^{(\particleindex)} \in \left ( \mathbf{L}_{k},\mathbf{U}_{k} \right ) \right \}$,

\State update estimate of normalizing constant $ \widehat{C}_{k} = \widehat{C}_{k-1} \times \sum_{\particleindex=1}^{\nparticles} w_{k}^{(\particleindex)}$,

\State compute normalized weights $ W_{k}^{(\particleindex)} \propto w_{k}^{(\particleindex)}$,

\State resample $ \left \{ W_{n}^{(\particleindex)}, X_{1:n}^{(\particleindex)} \right \}_{\particleindex=1}^{\nparticles}$ to obtain $\nparticles$ equally-weighted particles $\left \{ 1/\nparticles, \overline{X}_{1:n}^{(\particleindex)} \right \}_{\particleindex=1}^{\nparticles}$ and set $ \left \{ \overline{W}_{n}^{(\particleindex)}, \overline{X}_{1:n}^{(\particleindex)} \right \}_{\particleindex=1}^{\nparticles} \leftarrow \left \{ 1/\nparticles, \overline{X}_{1:n}^{(\particleindex)} \right \}_{\particleindex=1}^{\nparticles}$.

\EndIf

\EndFor

\end{algorithmic}
\end{algorithm}

The estimated price is
$$ \widehat{Q}_{D} = \widehat{C}_{\ntimepoints} \times \sum_{\particleindex=1}^{\nparticles} \overline{W}_{\ntimepoints}^{(\particleindex)} H \left ( \overline{\logassetvalue}_{N}^{(\particleindex)} \right ).  $$

Moreover, since essentially only the normalizing constant is being computed, there is no issue of path degeneracy \footnote{Path degeneracy is when repeated resampling steps lead to many multiple copies of the same particle $X_{1:N}$. This causes estimates based on the entire paths being unreliable.}. Resampling paths outside the barriers at the monitoring times from paths which are inside the barriers improves the efficiency of the estimator with respect to the standard MC estimator.
It is remarked that for constant volatility, in a Black-Scholes context, one can sample the price to ensure, in one step, that the process always survives; see \cite{glass_staum}. \\

If, however, very few particles satisfy the barrier condition at time $T_{i}$, then this estimate will also have a high variance. This can happen for example if: the dimension $d$ is high; the barrier condition is a narrow one, i.e, $\mathbf{L}_{i}$ and $\mathbf{U}_{i}$ are close to each other; the volatility $\boldsymbol\sigma$ is high; the time intervals $(T_{i-1},T_{i})$ are large. Our goal is to introduce a sequence of positive weighting functions so that they resolve this issue. In order to do so, we write 

\begin{eqnarray}
Q_{D} & = & \EE \bigg[ h_{0}(\logassetvalue_{0}) \frac{ h_{1}(\logassetvalue_{1}) } { h_{0}(\logassetvalue_{0}) } \times \cdots \frac{ h_{T_{1}}(\logassetvalue_{T_{1}}) }{ h_{T_{1}-1}(\logassetvalue_{T_{1}-1}) } \times h_{T_{1}+1} (\logassetvalue_{T_{1}+1}) \times  \frac{ h_{T_{1}+2}(\logassetvalue_{T_{1}+2}) } { h_{T_{1}+1}(\logassetvalue_{T_{1}+1}) } \times \cdots \nonumber \\ 
& & \hspace{0.2in} \times \frac{ h_{T_{\ntimepoints-1}}(\logassetvalue_{T_{\ntimepoints-1}}) } { h_{T_{\ntimepoints-1}-1}(\logassetvalue_{T_{\ntimepoints-1}-1}) } \times h_{T_{\ntimepoints-1}+1} (\logassetvalue_{T_{\ntimepoints-1}+1}) \times  \frac{ h_{T_{\ntimepoints-1}+2}(\logassetvalue_{T_{\ntimepoints-1}+2}) } { h_{T_{\ntimepoints-1}+1}(\logassetvalue_{T_{\ntimepoints-1}+1}) }  \times \cdots \times  \frac{ h_{T_{\ntimepoints}}(\logassetvalue_{T_{\ntimepoints}}) } { h_{T_{\ntimepoints}-1}(\logassetvalue_{T_{\ntimepoints}-1}) } \bigg] \label{eq:h_n} \nonumber \\
& = & h_{0}(\logassetvalue_{0}) \times \E{ \prod_{n=1}^{N} G_{n}(X_{n}) }, \label{eq:price_potential_functions}
\end{eqnarray} 
where
\begin{eqnarray}
G_{n}(X_{n}) & = & G_{n}(\logassetvalue_{n-1},\logassetvalue_{n}) = \frac{h_{n}(\logassetvalue_{n}) }{ h_{n-1}(\logassetvalue_{n-1}) }, \textrm{ } n \notin  \{ T_{1}+1, T_{2}+1, \ldots, T_{\ntimepoints-1}+1 \}, \nonumber \\
G_{T_{i}+1}(X_{T_{i}+1}) & = & G_{T_{i}+1}(\logassetvalue_{T_{i}},\logassetvalue_{T_{i}+1}) = h_{T_{i}+1}(\logassetvalue_{T_{i}+1}) \hspace{0.1in} \textrm{for } i = 1, 2, \ldots, \ntimepoints-1, \textrm{ and} \nonumber \\
h_{T_{i}}(\logassetvalue_{T_{i}}) & = & \mathbf{1} \{ \logassetvalue_{T_{i}} \in (\mathbf{L}_{i}, \mathbf{U}_{i}) \} \hspace{0.1in} \textrm{for } i = 1, 2, \ldots, \ntimepoints. \nonumber
\end{eqnarray}

The $h_{n}$'s, $n \notin \{T_{1}, T_{2}, \ldots, T_{\ntimepoints} \}$, are the sequence of positive weighting functions. In Algorithm \ref{algo:SMC_barrier}, they were simply 1. We attempt to choose them more carefully in order to approximate the optimal importance sampling densities. The algorithm is Algorithm \ref{algo:SMC_barrier_wt_fn}. \\

\begin{algorithm}
\caption{SMC for option pricing with weighting functions}\label{algo:SMC_barrier_wt_fn}
\begin{algorithmic}[1]

\State Set initial weights $\overline{W}_{0}^{(\particleindex)} = 1/ \nparticles$ and initial estimate of normalizing constant $\widehat{C}_{0} = 1$.

\For {$n=1$}

\State Sample $X_{1}^{(\particleindex)} \sim q_{1}(x_{1})$,

\State compute unnormalized weights
$$w_{1}^{(\particleindex)} = W_{0}^{(\particleindex)} \times G_{1} \left ( X_{1}^{(\particleindex)} \right ), $$

\State update estimate of normalizing constant $ \widehat{C}_{1} = \widehat{C}_{0} \times \sum_{\particleindex=1}^{\nparticles} w_{1}^{(\particleindex)}$,

\State compute normalized weights $ W_{1}^{(\particleindex)} \propto w_{1}^{(\particleindex)} $,

\If {resampling criterion is satisfied} 

\State resample $ \left \{ W_{1}^{(\particleindex)}, X_{1}^{(\particleindex)} \right \}_{\particleindex=1}^{\nparticles}$ to obtain $\nparticles$ equally-weighted particles $\left \{ 1/\nparticles, \overline{X}_{1}^{(\particleindex)} \right \}_{\particleindex=1}^{\nparticles}$ and set $ \left \{ \overline{W}_{1}^{(\particleindex)}, \overline{X}_{1}^{(\particleindex)} \right \}_{\particleindex=1}^{\nparticles} \leftarrow \left \{ 1/\nparticles, \overline{X}_{1}^{(\particleindex)} \right \}_{\particleindex=1}^{\nparticles}$,

\Else 

set $ \left \{ \overline{W}_{1}^{(\particleindex)}, \overline{X}_{1}^{(\particleindex)} \right \}_{\particleindex=1}^{\nparticles} \leftarrow \left \{ W_{1}^{(\particleindex)}, X_{1}^{(\particleindex)} \right \}_{\particleindex=1}^{\nparticles}$.

\EndIf

\EndFor

\For {$n \geq 2$}

\State Sample $X_{n}^{(\particleindex)} \sim q_{n}\left ( x_{n}|\overline{X}_{n-1}^{(\particleindex)} \right )$ and set $ X_{1:n}^{(\particleindex)} \leftarrow \left ( \overline{X}_{1:n-1}^{(\particleindex)}, X_{n}^{(\particleindex)} \right )$,

\State compute unnormalized weights
$$w_{n}^{(\particleindex)} = W_{n-1}^{(\particleindex)} \times G_{n} \left ( \mathbf{X}_{n}^{(\particleindex)} \right ), $$

\State update estimate of normalizing constant $ \widehat{C}_{n} = \widehat{C}_{n-1} \times \sum_{\particleindex=1}^{\nparticles} w_{n}^{(\particleindex)}$,

\State compute normalized weights $ W_{n}^{(\particleindex)} \propto w_{n}^{(\particleindex)} $,

\If {resampling criterion is satisfied} 

\State resample $ \left \{ W_{n}^{(\particleindex)}, X_{1:n}^{(\particleindex)} \right \}_{\particleindex=1}^{\nparticles}$ to obtain $\nparticles$ equally-weighted particles $\left \{ 1/\nparticles, \overline{X}_{1:n}^{(\particleindex)} \right \}_{\particleindex=1}^{\nparticles}$ and set $ \left \{ \overline{W}_{n}^{(\particleindex)}, \overline{X}_{1:n}^{(\particleindex)} \right \}_{\particleindex=1}^{\nparticles} \leftarrow \left \{ 1/\nparticles, \overline{X}_{1:n}^{(\particleindex)} \right \}_{\particleindex=1}^{\nparticles}$,

\Else 

set $ \left \{ \overline{W}_{n}^{(\particleindex)}, \overline{X}_{1:n}^{(\particleindex)} \right \}_{\particleindex=1}^{\nparticles} \leftarrow \left \{ W_{n}^{(\particleindex)}, X_{1:n}^{(\particleindex)} \right \}_{\particleindex=1}^{\nparticles}$.

\EndIf

\EndFor

\end{algorithmic}
\end{algorithm}

The estimated price is 
\begin{equation} \label{eq:nb_SMC_price}
\widehat{Q}_{D} = h_{0}(\logassetvalue_{0}) \times \widehat{C}_{N} \sum_{\particleindex=1}^{M} \overline{W}_{N}^{(\particleindex)} H \left ( \overline{\mathbf{S}}_{N}^{(\particleindex)} \right ).
\end{equation} 

Path degeneracy is again not an issue because we are still essentially estimating the normalizing constant. This estimate is unbiased and a proof is provided in the appendix. \\

Since $h_{T_{i}}$ is the indicator of $\logassetvalue_{T_{i}}$ being inside the barriers at time $T_{i}$, paths which are outside the barriers at the monitoring times are discarded with probability 1 in this case as well. However unlike in Algorithm \ref{algo:SMC_barrier}, here we seek to give higher weights to particles which we think have a higher chance of being inside the barriers at the monitoring times. What is being sought while using the weighting functions is an approximation to the optimal importance sampling density, that is, the density of $\mathbf{S}_{t}$ (at time $t$) conditional on it surviving at times $T_{i}$, $i = 1, 2, \ldots, p$. In the case of the barrier options being considered, this corresponds to $\mathbf{S}_{T_{i}} \in ( \mathbf{L}_{i}, \mathbf{U}_{i})$. This is aachieved by giving higher weights to particles which have a higher chance of surviving at times $T_{i}$. For example, particles which are far away from the barriers have a lower chance of survival than particles which are closer to the barriers. This is the intuitive idea behind our choice of weighting functions, and this is illustrated in Section \ref{numerics:barrier}.

\subsection{TARNs}

We consider a TARN based on a single underlying asset. Since the main problem arises when the function $f$ is discontinuous, we consider a discontinuous $f$ to illustrate our methods. Let

\[f(r) = \left\{ 
  \begin{array}{l l}
    2(r-110) + 20 & \quad \textrm{ if } r > 110 \\
   2(80-r) + 20 & \quad \textrm{ if } r < 90 \\
   -20 & \quad \textrm{ if } 90 \leq r \leq 110
  \end{array} \right.\]
\\
This function has two big jumps at $90$ and $110$. The  negative and positive cashflow cutoffs are $\Gamma_{L} = 100$ and $\Gamma_{G} = 200$. We recall that we had earlier assumed the interest rate is $0$. This is now justified. The main reason why standard MC cannot be used efficiently in this scenario is as follows. Using MC, most of the particles stay inside $(90,110)$ for the first five fixing dates. This leads to the contribution of the particle in the MC estimate being $-100$. However, an occasional particle escapes $(90,110)$ within the first five fixing dates and contributes a value that is significantly different from $-100$ because of the big jumps in $f$. This causes the variance of the MC estimate to be high. Even if the interest rate was positive, this difficulty would still remain. Therefore for simplicity we assume the interest rate to be $0$. \\

In order to go back to the previous notation, let $S = \log(R) \in \RR$ and define $g(s) = f(e^{s})$. Then let 

\begin{equation} \nonumber 
g(S_{1:N}) := 100 + \sum_{i=1}^{\tau} g(S_{T_{i}})
\end{equation}
be the new payoff function, where we have written $S_{1:N}$ in place of $(S_{t_{1}}, S_{t_{2}}, \ldots, S_{t_{N}})$. By defining $g$ in this way, the problem has been transformed into the format that was being using before. In this case,
\begin{eqnarray}
\mathcal{S} & = & \RR^{N}, \nonumber \\
\mathcal{S}_{1} & = & \left \{ s_{1:N} \in \RR^{N} : L_{5} := \sum_{i=1}^{5} g^{-}(s_{T_{i}}) = 100 \right \}, \nonumber \\
\mathcal{S}_{2} & = & \left \{ s_{1:N} \in \RR^{N} : L_{5} := \sum_{i=1}^{5} g^{-}(s_{T_{i}}) < 100 \right \}. \nonumber
\end{eqnarray}
A particle lies in $\mathcal{S}_{1}$ if (and only if) it stays within $(90,110)$ for the first five fixing dates. $\mathbb{P}(\mathcal{S}_{1}) \gg \mathbb{P} ( \mathcal{S}_{2} )$ if, for example, the volatility is low or the time intervals $T_{i} - T_{i-1}$ are small. It is noted that this is the opposite of the barrier option case (in which we considered the time intervals and volatility to be large). We again consider a sequence of positive weighting functions $h_{1}, h_{2}, \ldots, h_{T_{5}}$ and write

\begin{equation} \nonumber
g(s_{1:N}) = \left [ \prod_{n=1}^{T_{5}} \frac{ h_{n}(s_{n}) }{ h_{n-1}(s_{n-1}) } \right ] \times \frac{ g(s_{1:N}) } { h_{T_{5}} (s_{T_{5}}) },
\end{equation} 
where 
$ h_{0}(s_{0}) \equiv 1 $; we are basically doing the same thing as in the barrier option case. The goal is again to guide particles towards being in $\mathcal{S}_{2}$, and we show that this can be achieved through the usage of some simple weighting functions in Algorithm \ref{algo:SMC_barrier_wt_fn}.

\section{Numerical results}\label{sec:numerics}

In this section we demonstrate numerically the benefit of using weighting functions. In order to compare the standard deviations of Algorithms \ref{algo:SMC_barrier} and \ref{algo:SMC_barrier_wt_fn}, we run them 100 times with $100000$ particles in each run. We then look at the standard deviations of the 100 estimates and report the relative standard deviations (the ratio of the standard deviations).

\subsection{Barrier options} \label{numerics:barrier}

We consider a barrier option whose asset values evolve independently of each other. This translates to $\Sigma = I_{d}$. The option type is call. We assume that we can simulate forward one day at a time. A common monitoring strategy is to monitor the underlying assets after every $k$ days for a total of $\ntimepoints$ time periods. In that case, $T_{i} = ik$ for $i = 1, \ldots, \ntimepoints$ and $N = \ntimepoints k$. We choose $k = 540$ and $\ntimepoints=1$. Algorithm \ref{algo:SMC_barrier} resamples at the end of $\ntimepoints$ days and so resets the system of particles by choosing all resampled particles being inside the barriers. Therefore the gain that we expect by using Algorithm \ref{algo:SMC_barrier_wt_fn} can only be before the system is reset and this is why we choose $\ntimepoints = 1$. In what follows, we refer to Algorithm \ref{algo:SMC_barrier} as `MC' and Algorithm \ref{algo:SMC_barrier_wt_fn} as `SMC'. \\

The algorithms are run on different values of the dimension $d$ and the volatility $\boldsymbol\sigma = (\sigma_{1}, \ldots, \sigma_{d})$.  Recalling \eqref{eq:h_n}, the targeted density at any time $n$ is proportional to $h_{n}(\logassetvaluesmall_{n}) p_{n}(\logassetvaluesmall_{1:n})$, where $p_{n}(\logassetvaluesmall_{1:n})$ denotes the density of $\logassetvaluesmall_{1:n}$. This implies that the targeted marginal density at time $n$ is proportional to $h_{n}(\logassetvaluesmall_{n}) p_{n}(\logassetvaluesmall_{n})$, where $p_{n}(\logassetvaluesmall_{n})$ denotes the marginal density of $\logassetvaluesmall_{1:N}$ at time $n$. We denote the targeted density at time $n$ by $\widetilde{p}_{n}(\logassetvaluesmall_{n})$.

\subsubsection{Constant volatility model}

Since our basket consists of $d$ independent assets, we choose the marginal targeted densities (at different time points) to be the product of $d$ (unnormalized) densities. That is, $\widetilde{p}_{n} : \RR^{d} \rightarrow \RR$ is such that $ \widetilde{p}_{n}(\logassetvaluesmall_{n}) = \prod_{j=1}^{d} \widetilde{p}_{n,j}(s_{n,j})$, where $\logassetvaluesmall_{n} = (s_{n,1}, \ldots, s_{n,d})$. In the constant volatility case, the marginal density of $\logassetvalue_{1:N}$ at time $n$ is known and is denoted by $p_{n}(\logassetvaluesmall_{n}) = \prod_{j=1}^{d} p_{n,j}(s_{n,j})$; this is simply a product of $d$ Gaussians. Therefore the weighting functions are $ h_{n}(\logassetvaluesmall_{n}) = \prod_{j=1}^{d} \left [ \widetilde{p}_{n,j}(s_{n,j}) / p_{n,j}(s_{n,k}) \right ] $. \\

We choose $\sigma_{1} = \cdots = \sigma_{d} = \sigma$. Since our goal is to push particles towards being inside the barriers at time $k$, we consider a (one dimensional) Brownian bridge with volatility $\sigma$ tied down at $B_{0,j} = S_{0,j}$ and $B_{k,j} = K_{j} := (L_{j}+U_{j})/2$. However, since the variance of the Brownian bridge goes to 0 as the time approaches $k$, we cannot directly use it as it would lead to a high variance. So we add $0.2 \sigma$ to the standard deviation of the Brownian bridge. It is remarked that $0.2 \sigma$ is somewhat an arbitrary choice. This is not necessarily the best choice, and we do not claim so; this choice just serves to illustrate the benefit of using weighting functions. We have noticed gains even while using other values, all we need to do is to ensure that the targeted densities don't tend towards being degenerate. In practice, we introduce the weighting functions only after time $n = 2k/3$ because we want to let particles initially explore the space and then give higher weights to particles which are more likely to be inside the barriers at time $k$. \\

\begin{figure}[p]
    \centering
    \includegraphics[width=1\textwidth]{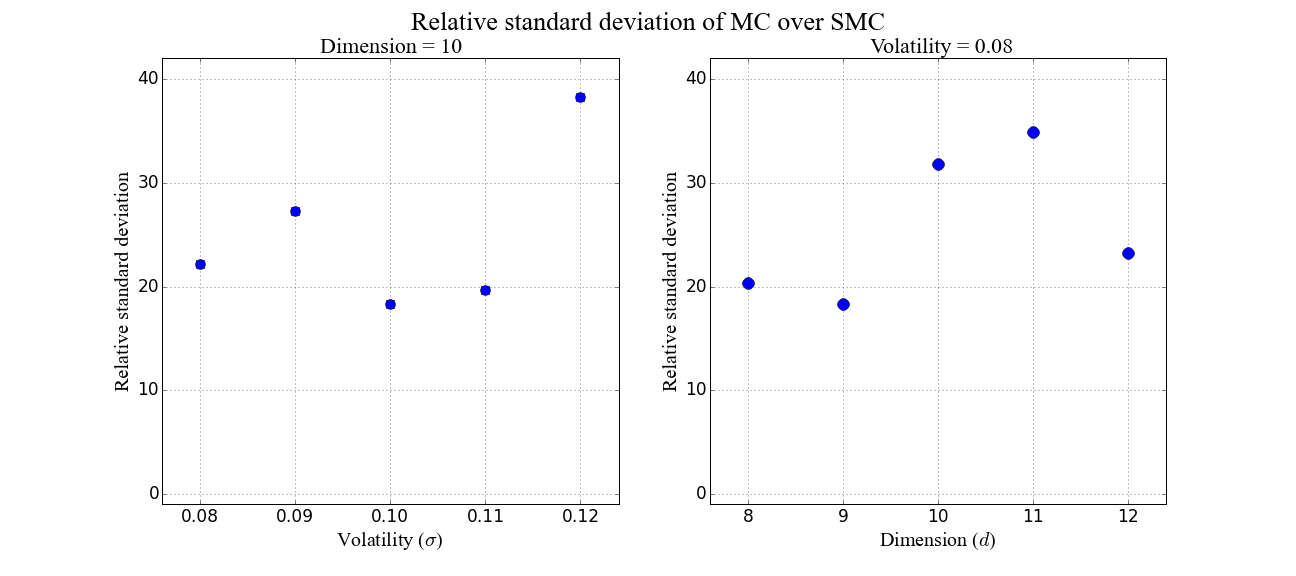}
    \caption{Relative standard deviations of MC over SMC (Algorithm \ref{algo:SMC_barrier_wt_fn}) for a barrier option with $m = 1$ and $k = 540$ for a constant volatility model when we target a Brownian bridge; in the left figure, the dimension $d = 10$; in the right figure, the volatility $\sigma = 0.08$.}
    \label{fig:fig2}
\end{figure} 

The results are in Figure \ref{fig:fig2}. We observe a gain in the standard deviation by using these functions. In general, if the volatilities, initial prices and strike prices were different for the different chains, we could have considered a product of $d$ dissimilar Brownian bridge densities as the weighting functions. \\

Since the actual goal is to estimate the optimal importance density, we try to do it in our simple setup to see how much better it does. The marginal densities of a particle given that it survives at time $k$ can be approximated by Gaussian densities. In fact, this is what we were doing earlier when we choose the targeted distribution to be a Brownian bridge. Instead of working with a Brownian bridge, we can try to estimate the means and the variances of the Gaussian densities by first simulating particles in one dimension and then looking at the means and variance of those that survive. \\

In this regard, we simulate $M_{1} = 10000$ particles in one dimension with volatility $\sigma = 0.08$. We then look at the particles that survive at the end. Given the particles that survive, we look at their marginal means and variances. This approximates the marginal densities (in one dimension) of a particle given that it survives. Since we are in one dimension, the number of particles surviving is much higher than what we would have in higher dimensions (about 3900 out of 10000 in our simulations). This means that the probability of survival in one dimension is about $0.39$ for $\sigma = 0.08$, which for example in 10 dimensions would mean a probability of survival of $8.1 \times 10^{-5}$ (since the chains are independent). We call this the `optimal' target. \\

The results are in Figure \ref{fig:fig3}. We can observe an improvement against the standard MC method, relative to Figure \ref{fig:fig2}. However, the gains are not very drastic and serve to illustrate that even the simple intuitive use of a Brownian bridge does well. It is worth noting that even by choosing the volatility to be $0.08$ and running the chain in one dimension, we observe a gain for other values of the volatility as well. This suggests that even if the volatilities were different for the different chains, we could have simply run a chain in one dimension with a value of the volatility that is in the range of the volatilities and still get decent results. \\

\begin{figure}[p]
    \centering
    \includegraphics[width=1\textwidth]{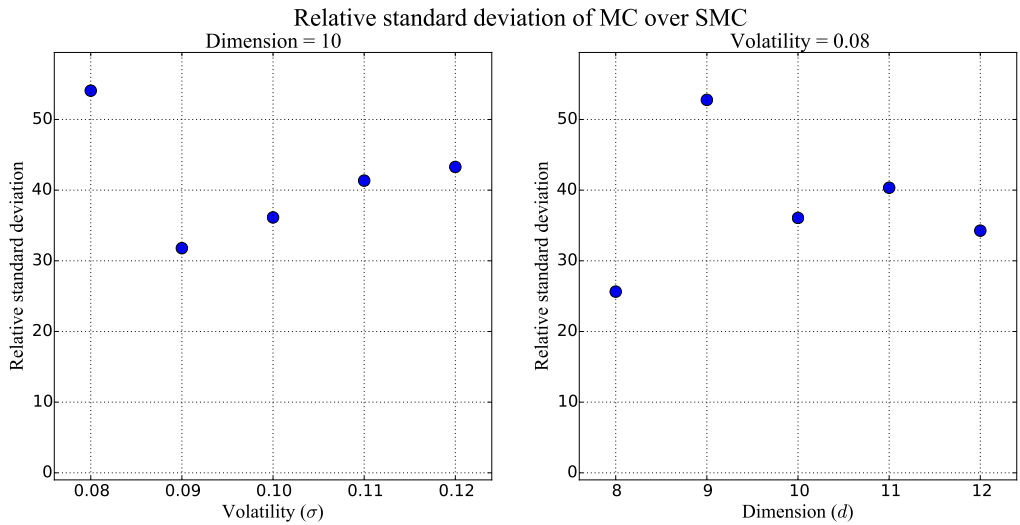}
    \caption{Relative standard deviations of MC over SMC (Algorithm \ref{algo:SMC_barrier_wt_fn}) for a barrier option with $m = 1$ and $k = 540$ for a constant volatility model when we target the `optimal' target; in the left figure, the dimension $d = 10$; in the right figure, the volatility $\sigma = 0.08$.}
    \label{fig:fig3}
\end{figure} 

A note on running times: we have observed that the running time for MC is less than three times that of SMC (even without having optimized the codes too much). Thus if we were to take running times into account, we could have run MC with $3\nparticles$ particles in the same time in which we run SMC with $\nparticles$ particles. The standard deviation for MC using $3\nparticles$ particles would be $\sqrt{3}$ times less than that of MC with $\nparticles$ particles. So even by taking running times into account, SMC outperforms MC. In the last example, in addition we run MC with $10000$ particles. We should take this time into account as well, but we run it for only one value of the volatility and it is fast. The reason we do not consider the running times in our figures is that we believe it might be possible to optimize the codes even further. The intention is to demonstrate the benefit of using weighting functions in SMC.

\subsubsection{Local volatility model}

Local volatility models are in practice used more frequently as they are typically more accurate. The local volatility function considered here is a linear interpolation between the values of the volatility $\sigma$ at $0.12, 0.11, 0.105, 0.101, 0.097, 0.093, 0.098, 0.10, 0.105, 0.11, 0.17$ and the values of the underlying $R$ at $10^{-6}, 60, \\ 70, 80, 90, 100, 110, 120, 130, 140, 10^{6}$. In this case, we do not know the marginal densities of $\logassetvalue_{1:N}$. For simplicity, we still assume that the volatility functions are the same for all the chains and approximate it in the following way. Consider a single chain. If we knew  $\EE S_{1}$, $\EE S_{2}, \ldots$, then we could approximate $\sigma^{2} (S_{n})$ by $\sigma^{2} (\EE S_{n})$. We approximate $\EE S_{n}$ iteratively as follows (assuming that the time step of the discretization $\delta t_{n}$ is constant):

\begin{eqnarray}
\CE{ S_{n} }{ S_{n-1} } & = & S_{n-1} - \frac{1}{2} \sigma^{2}(S_{n-1}) \delta t \nonumber \\
\Rightarrow \EE S_{n} & = & \EE S_{n-1} - \frac{1}{2} \delta t \E{ \sigma^{2}(S_{n-1})} \nonumber \\
& \approx & \EE S_{n-1} - \frac{1}{2} \delta t \sigma^{2} \left ( \EE S_{n-1} \right ); \label{eq:approx_expectation}
\end{eqnarray}
therefore we approximate $\EE S_{n}$ and $\sigma^{2}(\EE S_{n})$ iteratively. We first approximate $\EE S_{1}$ by \eqref{eq:approx_expectation}; given this estimate, we approximate $\EE S_{2}$ by \eqref{eq:approx_expectation} and keep doing this. We approximate $p_{n}(\logassetvalue_{n})$ by
$$ \widehat{p}_{n} (\logassetvalue_{n}) = \prod_{j=1}^{d} \phi \left ( s_{n,j} \bigg| \EE S_{n-1} - \frac{1}{2} \sigma^{2} ( \EE S_{n-1} ) \delta t, \sigma^{2}(\EE S_{n-1} ) n \delta t \right ) $$
iteratively. 

\begin{enumerate}

\item We begin with a Brownian bridge weighting function using the previously estimated values of $\sigma(\EE S_{n})$ in the target density and add $0.2 \sigma(\EE S_{n})$ to the standard deviation as before.

\item As in the constant volatility model, we simulate $M_{1} = 10000$ particles (for the local volatility model) in one dimension and look at the marginal means and variance of the particles that survive. We call this the `optimal' target and target it in the SMC algorithm. 

\end{enumerate}

The results are in figure \ref{fig:barrier_lv}. We observe a significant gain, which as expected, gets more significant as the dimension increases. This is because the probability of a particle surviving gets smaller as the dimension increases. The gains are more significant when we target the `optimal'. As far as running times are concerned, the observations are the same as in the case of the constant volatility model. Thus Algorithm \ref{algo:SMC_barrier_wt_fn} outperforms Algorithm \ref{algo:SMC_barrier} even upon taking running times into consideration.

\begin{figure}[p]
    \centering
    \includegraphics[width=1\textwidth]{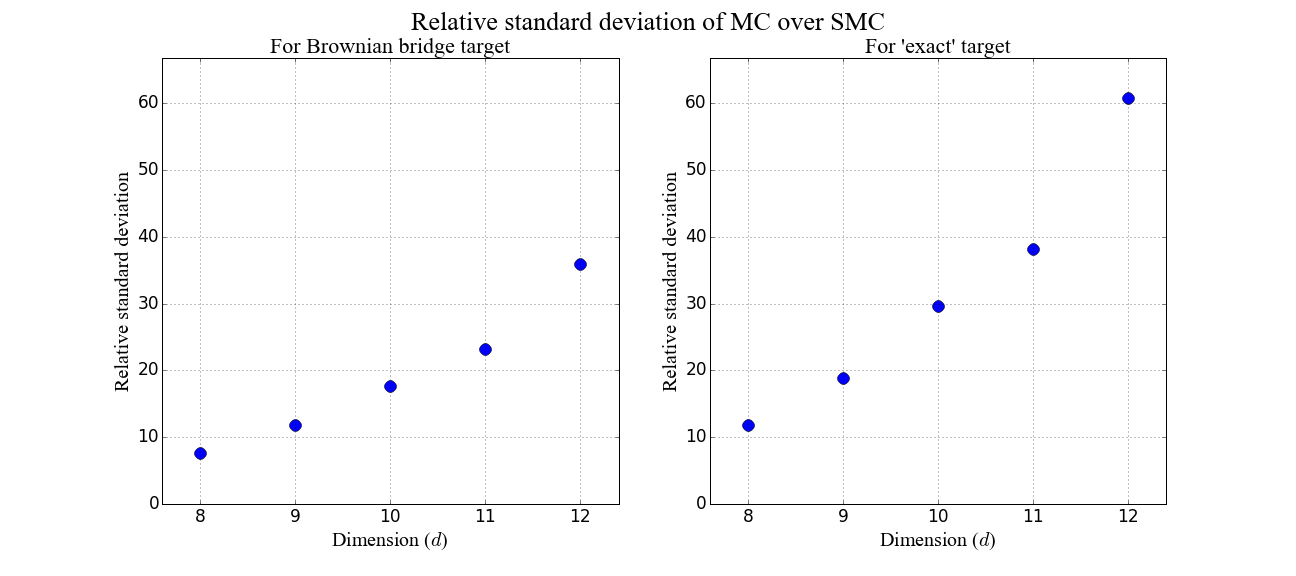}
    \caption{Relative standard deviations of MC over SMC (Algorithm \ref{algo:SMC_barrier_wt_fn}) for a barrier option with $m = 1$ and $k = 540$ for local volatility model; the left figure is when we target a Brownian bridge; the right figure is when we target the `optimal'.}
    \label{fig:barrier_lv}
\end{figure} 

\subsection{TARNs} \label{numerics:TARN}

We again assume that we can simulate forward one day at a time. We take the number of fixing dates $\ntimepoints = 24$ and the interval between fixing dates $k = 30$ days. This roughly corresponds to a TARN where there is a cash flow at the end of every month for a maximum of two years. We consider different values of the volatility. A particle leads to a zero payoff if it stays within $(-90,110)$ for the first five months and to a positive payoff if it escapes by the end of the fourth month. We again compare standard MC and Algorithm \ref{algo:SMC_barrier_wt_fn} (referring to this as `SMC'). We note that contrary to barrier options, in this case we expect SMC to do better than usual MC for \textit{low} values of the volatility. This is because in this case the region $\mathcal{S}_{2}$ has low probability for lower values of the volatility unlike in the barrier option case where $\mathcal{S}_{2}$ had low probability for higher values of the volatility. This is why we restrict ourselves to lower values of the volaitility in this section than in Section \ref{numerics:barrier}. 

\subsubsection{Constant volatility model}

When the underlying follows Black-Scholes dynamics with a constant volatility $\sigma$, we can directly simulate forward $k$ days at a time using the Euler-Maruyama discretization. We run the algorithms different values of the volatility and report the results. We try three (increasingly more sophisticated, yet intuitive) weighting functions to try and get particles to escape. \\

In the most naive version, we choose $h_{n}(s_{1:n})$ to simply be $(s_{n} - S_{0})^{2}$. This just gives higher weights to particles that are further away from $S_{0}$ at the end of the $n$-th month. We obtain the results in Figure \ref{fig:tarn_cv}. \\
 
\begin{figure}[p]
    \centering
    \includegraphics[width=1\textwidth]{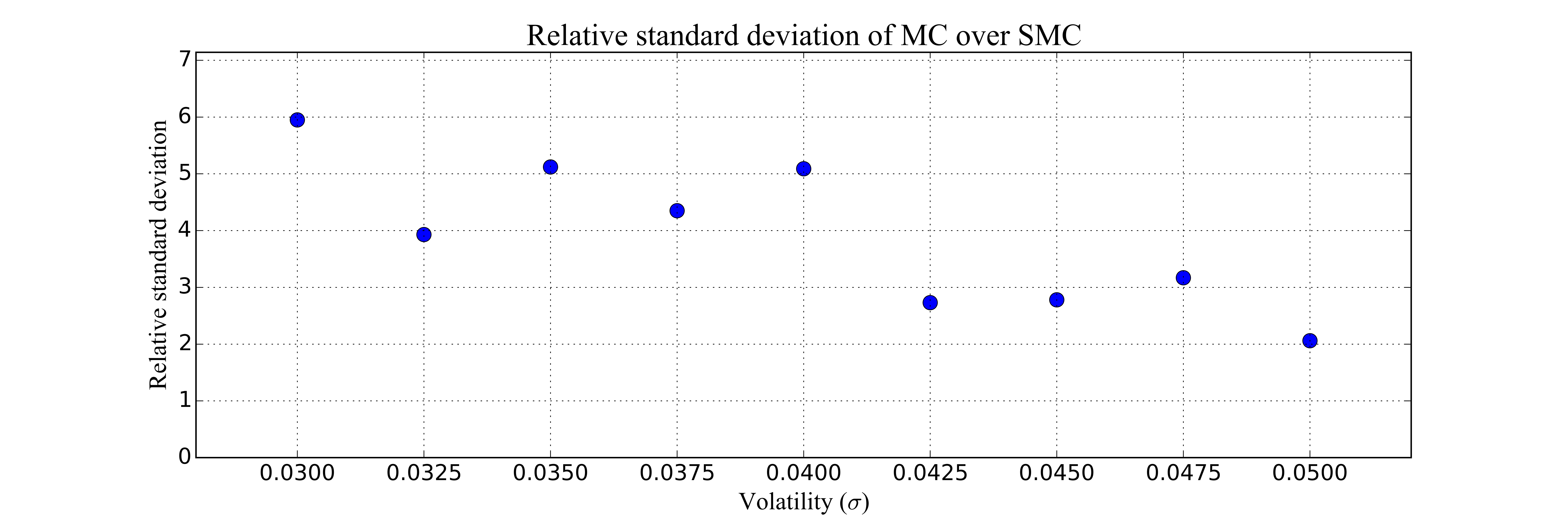}
    \caption{Relative standard deviations of MC over SMC (Algorithm \ref{algo:SMC_barrier_wt_fn}) using most naive weighting functions for TARN in the constant volatility case.}
    \label{fig:tarn_cv}
\end{figure}

We then notice that the targeted density at time $n$ is proportional $h_{n}(s_{1:n}) p(s_{1:n})$ for $1 \leq n \leq 5$. If we choose $h_{n}(s_{1:n})$ to be of the form $h_{n}(s_{n})$, then the marginal of the targeted density at time $n$ is proportional to $h_{n}(s_{n}) p_{n}(s_{n})$, where $p_{n}(s_{n})$ is the marginal density of $S_{1:N}$ at time $n$. Since we target $(s_{n} - S_{0})^{2}$, we choose $ h_{n}(s_{n}) = (s_{n} - S_{0} )^{2} / p_{n}(s_{n}) $ and this leads to the results in Figure \ref{fig:tarn_cv1}. It is seen that SMC with weighting functions does much better in most cases, except in the last few cases. The last few cases correspond to higher volatilities, in which case the probability of a particle escaping increases. SMC does better when the probability of a particle escaping is low, which is what we wanted. \\

\begin{figure}[p]
    \centering
    \includegraphics[width=1\textwidth]{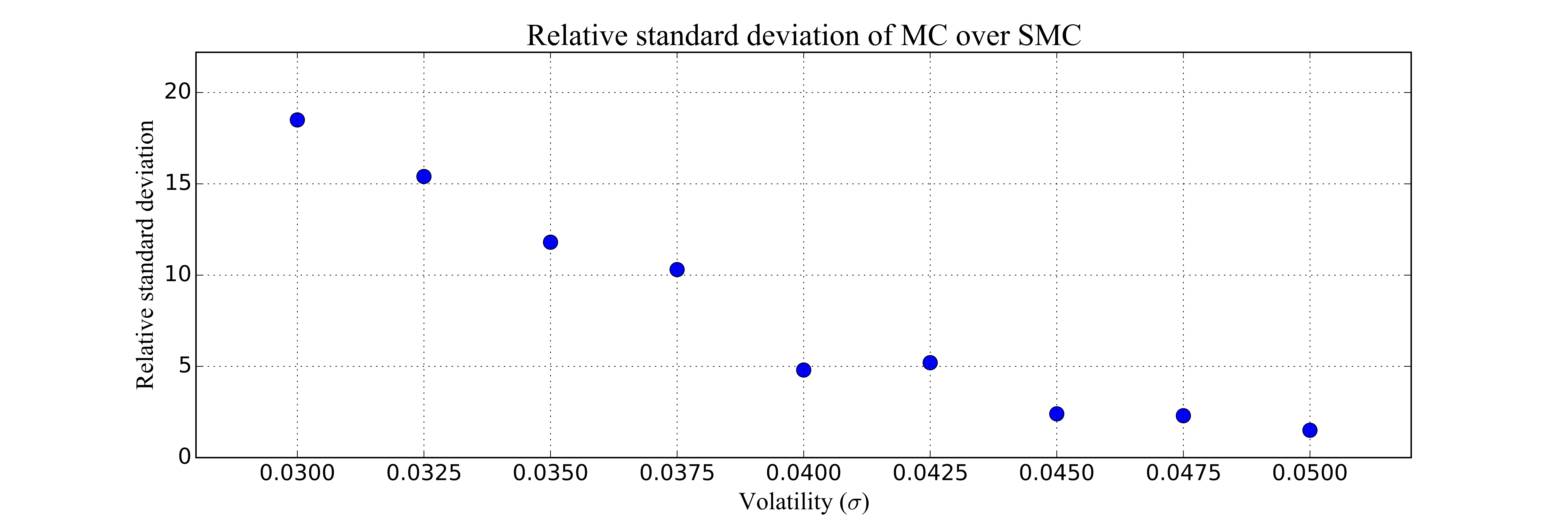}
    \caption{Relative standard deviations of MC over SMC (Algorithm \ref{algo:SMC_barrier_wt_fn}) using simple weighting functions for TARN in the constant volatility case.}
    \label{fig:tarn_cv1}
\end{figure}

Looking into the problem further, we observe that particles which escape can do so in one of two ways: either to the left or to the right. We fix a value of $\sigma = 0.05$ and run $10000$ copies of the chain. Our goal is to look at the exact marginal distribution of particles of particles which have escaped. We observe that approximately $20\%$ of particles escape to the left and $80 \%$ escape to the right. This is of course for a fixed $\sigma$, but we use these marginal approximations as our target (we use a mixture of normals) and run the SMC algorithm for different values of $\sigma$. The results are in  Figure \ref{fig:tarn_cv2}. Significant gains are observed in this case as well, and are more than the gains before. In this case as well, we are trying to approximate the optimal importance density. \\

\begin{figure}[p]
    \centering
    \includegraphics[width=1\textwidth]{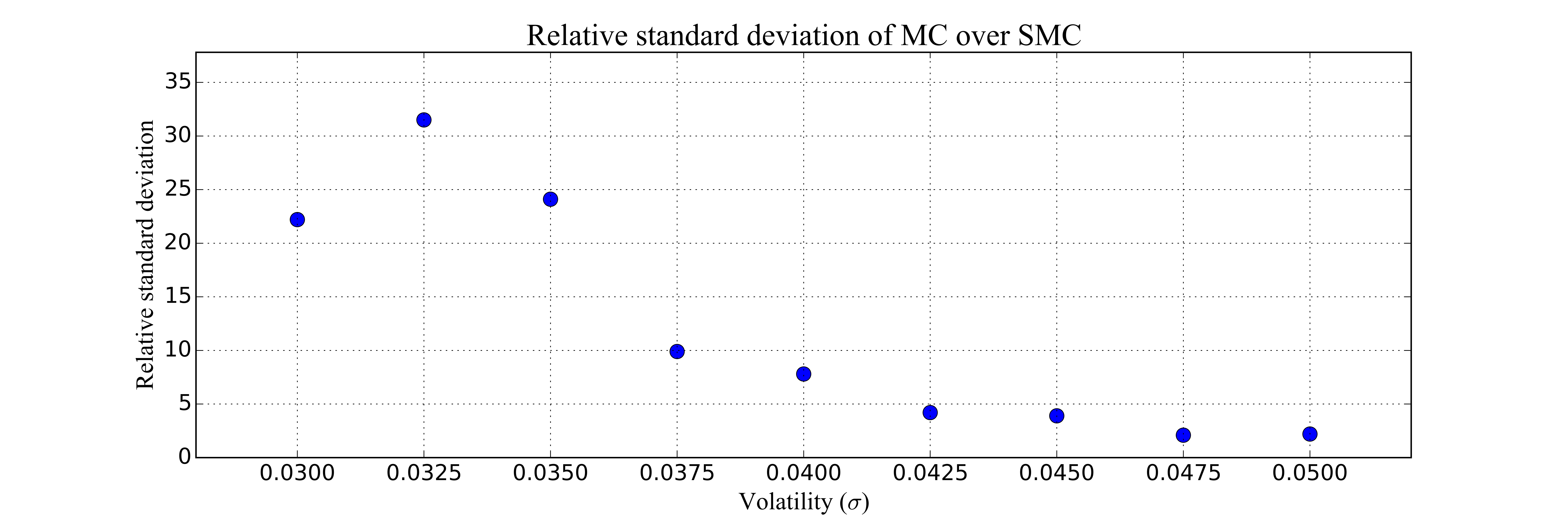}
    \caption{Relative standard deviations of MC over SMC (Algorithm \ref{algo:SMC_barrier_wt_fn}) using a mixture of normal densities as the target for TARN in the constant volatility case.}
    \label{fig:tarn_cv2}
\end{figure}

In all of the above three examples in this section, we have observed that the running time for SMC is less than twice that of MC. We can thus have similar conclusions as in the case of barrier options, which shows that SMC outperforms MC even upon taking running times into consideration. 

\subsubsection{Local volatility model} 

In this case, the volatility $\sigma$ is a function of the price of the underlying $R_{t}$. We cannot simulate forward $k$ days at a time any more and we simulate forward one day at a time now. In this case the weighting functions are introduced each day for the first 5 time periods. \\

We consider a volatility function which is minimum at $100$ and increases on either side of $100$. The value at $100$ is $0.036$. The values in the $(90,110)$ are less than $0.04$ and are higher outside it. We choose this because a low value of $\sigma$ would mean less particles escaping, but for the ones that do escape we let explore the space more freely. The local volatility function is a linear interpolation between the values of the volatility $\sigma$ at $0.055 ,0.051, 0.045, 0.041, 0.037, 0.035, 0.038,0.04, 0.045, 0.05, 0.055$ and the values of the underlying $R$ at $10^{-6}, 60, 90, 93, 98, 100, 103, 107, 110, 140, 10^{6}$. \\

We start off with the simplest weighting function $h_{n}(s_{1:n}) = (s_{n} - S_{0})^{2}$. In this case, we observe that Algorithm \ref{algo:SMC_barrier_wt_fn} has a standard deviation which is 1.83 times lower than that of MC. \\

Since the marginal density of $S_{n}$ is unknown, we choose a crude approximation of the marginal and use it. We choose 
$$ \widehat{p}_{n}(s_{n}) = \phi \left ( s_{n} \bigg| S_{0} + \left( \mu - \frac{1}{2} 0.04^{2} \right ) \delta t_{n}, 0.04^{2} \delta t_{0} \right ) $$ 
to be the approximation. This is the density of $s_{n}$ if the volatility had been a constant $0.04$. We run $10000$ copies of the chain and look at the ones which escape towards the left and towards the right. We then choose a mixture of normals (in this case, the proportions $0.3$ and $0.7$) and repeat the experiments. We observe that SMC has a standard deviation which is 2.33 times less than that of MC. \\ 

The running times for SMC are twice that of usual MC in this case. This implies that while there are some gains to be achieved in using SMC in this example, they are not as significant as before. This is because of the fact that the local volatility function is low (in the region of around $0.035$) only for a few values of the underlying, which suggests that one shouldd actually give even higher weights to force particles to escape. Nevertheless, this does serve to illustrate the potential benefits of using SMC over MC.

\section{Conclusion}

We have formally presented the idea of weighting functions in SMC. The main idea is to try to estimate the sequence of optimal importance densities using a relatively heuristic technique of weighting functions, giving more weights to potentially favourable particles. We demonstrated this on two examples from finance, but this can be extended to other cases (for instance those in \cite{SMC_option_jasra, jasra_doucet}). We have also seen that as we get closer to the optimal importance density the gains keep increasing. However even using an approximation leads to significant gains. Since the idea is quite general, it can potentially be used in finance to price other kinds of path dependent options, for instance Asian options as in \cite{SMC_option_jasra}. It can also be used in the context of the estimation of marginal expectations w.r.t.~laws of jump-diffusions (and hence some control problems); see for instance \cite{jasra_doucet}. The basic notion of weighting functions is also applied in \cite{beskos1} for high-dimensional filtering problems. In summary the ideas in the article provide a simple formalisation of many that have previously appeared in the literature. \\

There are many extensions of the work in the article. For instance one can adaptively determine the potential functions in the sequence, albeit at the cost of statistical bias. Other ideas include further applications and adaptation to problems where the multi-level Monte Carlo method can be used (e.g.~\cite{beskos1}).

\appendix

\section{Unbiasedness of Estimate}

We provide a proof of the unbiasedness of the estimate when the resampling is done adaptively and multinomially. 

\subsection{Additional Definitions}

Let $\tau_{1} < \tau_{2} < \cdots < \tau_{r}$ be the successive resampling times and let $\widetilde{\tau}(n)$ be the most recent resampling time before time $n$, with $\widetilde{\tau}(n) = 0$ if no resampling has been carried out before time $n$. Also define $\tau_{0} = 0$ and $\tau_{r+1} = N$. Define 
$$ v_{n} ( x_{1:n}) = \prod_{t = \widetilde{\tau}(n) + 1}^{n} \alpha_{n}(x_{1:n}), \hspace{0.15in} V_{n}^{(\particleindex)} = \frac{ v_{n} \left ( X_{1:n}^{(\particleindex)} \right )} { \nparticles \overline{v}_{n} }, $$
where 
$$ \overline{v}_{n} = \frac{1}{\nparticles} \sum_{j=1}^{\nparticles} v_{n} \left ( X_{1:n}^{(j)} \right ). $$
If resampling is carried out at time $n$, then $\left \{ V_{n}^{(\particleindex)} \right \}_{\particleindex=1}^{\nparticles}$ are the resampling weights. 
\\ \\
For any test function $\psi : \mathcal{X}^{N} \rightarrow \RR$, we estimate $ \psi_{N} := \EE_{\pi_{N}} \sqBK{ \psi(X_{1:N}) }$ by $ \widehat{\psi}_{OR} = \sum_{\particleindex=1}^{\nparticles} V_{N}^{(\particleindex)} \psi \left ( X_{1:N}^{(\particleindex)} \right ) $ and we estimate the normalizng constant by $ Z_{N}^{M} = \prod_{s=1}^{r+1} \overline{v}_{\tau_{s}} $.
Define:
\begin{equation} 
\widetilde{H}_{\tau_{s}}^{(\particleindex)} = \overline{v}_{1} \cdots \overline{v}_{\tau_{s}} h_{\tau_{s}} \left ( X_{1:\tau_{s}}^{(\particleindex)} \right ), \hspace{0.15in} H_{\tau_{s}}^{(\particleindex)} =  \overline{v}_{1} \cdots \overline{v}_{\tau_{s}} h_{\tau_{s}} \left ( \overline{X}_{1:\tau_{s}}^{(\particleindex)} \right ), \nonumber
\end{equation}
where
\begin{equation}
h_{\tau_{s}}(x_{1:\tau_{s}} ) = \frac{1}{ \prod_{n=1}^{\tau_{s}} \alpha_{n} (x_{1:n} ) }. \nonumber
\end{equation}
\\
Observe that
\begin{eqnarray} 
\frac{ H^{(\particleindex)}_{\tau_{s-1}}}{\widetilde{H}_{\tau_{s}}^{(\particleindex)}} & = & \frac{1}{\overline{v}_{\tau_{s}}} \times \frac{ \prod_{n=1}^{\tau_{s}} \alpha_{n} \left ( X_{1:n}^{(\particleindex)} \right )} {  \prod_{n=1}^{\tau_{s-1}} \alpha_{n} \left ( \overline{X}_{1:n}^{(\particleindex)} \right )} \nonumber \\
& = & \frac{ v_{\tau_{s}} \left ( X_{1:\tau_{s}}^{(\particleindex)} \right ) } { \overline{v}_{\tau_{s}} } \nonumber \\ 
& = & \nparticles V_{\tau_{s}}^{(\particleindex)} \nonumber \\
\Rightarrow  \widetilde{H}_{\tau_{s}}^{(\particleindex)} V_{\tau_{s}}^{(\particleindex)} & = & \frac{1}{\nparticles} H^{(\particleindex)}_{\tau_{s-1}}. \label{eq:3.new}
\end{eqnarray}
Define the likelihod ratio as 
$$ L_{N} (x_{1:N}) = \frac{ \pi_{N}(x_{1:N}) } { \prod_{n=1}^{N} q_{n}(x_{n}|x_{1:n-1}) }, $$
and let 
\begin{eqnarray}
\widetilde{\psi}_{OR} & = & \frac{1}{\nparticles} \sum_{\particleindex=1}^{\nparticles} L_{N} \left ( X_{1:N}^{(\particleindex)} \right ) \psi \left ( X_{1:N}^{(\particleindex)} \right ) H_{\widetilde{\tau}(N)}^{(\particleindex)} \nonumber \\
& = & \frac{1}{\nparticles Z_{N}} \sum_{\particleindex=1}^{\nparticles} \left [ \frac{ \gamma_{N} \left ( X_{1:N}^{(\particleindex)} \right )}{ \prod_{n=1}^{N} q_{n} \left ( X_{n}^{(\particleindex)} | X_{1:n-1}^{(\particleindex)} \right ) } \psi \left ( X_{1:N}^{(\particleindex)} \right ) \frac{ \overline{v}_{1} \cdots \overline{v}_{\tau_{r}} } {\prod_{n=1}^{\tau_{r}} \alpha_{n} \left ( \overline{X}_{1:n}^{(\particleindex)} \right ) } \right ] \nonumber \\
& = & \frac{ \overline{v}_{1} \cdots \overline{v}_{\tau_{r}} \overline{v}_{\tau_{r+1}} }{Z_{N}} \sum_{\particleindex=1}^{\nparticles} \frac{ v_{N} \left ( X_{1:N}^{(\particleindex)}\right )} {\nparticles \overline{v}_{N} } }{  \psi \left ( X_{1:N}^{(\particleindex)} \right ) \nonumber \\
& = & \frac{ Z_{N}^{\nparticles} }{ Z_{N} } \widehat{\psi}_{OR}. \nonumber 
\end{eqnarray} 
The third equality is because $\tau_{r+1} = N$. 
Define $A_{\tau_{s}}^{(\particleindex)}$ iteratvely as $A_{\tau_{0}}^{(\particleindex)} =1$ and $A_{\tau_{s}}^{(\particleindex)} = A_{\tau_{s-1}}^{ \left ( I_{\tau_{s-1}}^{(\particleindex)} \right ) } $, where we recall that $I_{\tau_{s-1}}^{(\particleindex)}$ is the resampled index of the $\particleindex$-th particle at the $\tau_{s-1}$-th resampling step. Let
\begin{eqnarray}
\FF_{2t-1} & = & \sigma \left ( \left \{ X_{1}^{(\particleindex)}: 1 \leq \particleindex \leq \nparticles \right \} \bigcup \left \{ \left ( \overline{X}_{1:\tau_{s}}^{(\particleindex)}, X_{1:\tau_{s}+1}^{(\particleindex)}, A_{\tau_{s}}^{(\particleindex)} \right ) : 1 \leq s < t, 1 \leq \particleindex \leq \nparticles \right \} \right ) , \nonumber \\
\FF_{2t} & = & \sigma \left ( \FF_{2t-1} \bigcup \left \{ \left ( \overline{X}_{1:\tau_{t}}^{(\particleindex)}, A_{\tau_{t}}^{(\particleindex)} \right ) : 1 \leq \particleindex \leq \nparticles \right \} \right ); \nonumber 
\end{eqnarray}
these are the $\sigma$-fields generated by the random variables associated with the $M$ particles just before and just after the $t$-th resampling step respectively. Let $\widetilde{f}_{0} (\cdot) \equiv \psi_{N}$ and define for $1 \leq n \leq N$, 
\begin{equation} \label{eq:3.2}
\widetilde{f}_{n} ( x_{1:n}) = \EE_{q} \sqBK{ \psi ( \overline{X}_{1:N}) L_{N} ( \overline{X}_{1:N} ) | \overline{X}_{1:n} = x_{1:n} },
\end{equation}
where $\EE_{q}$ denotes expectation under the proposal density. Then $ \widetilde{f}_{n} ( x_{1:n}) = \EE_{q} \sqBK{ \widetilde{f}_{N} ( \overline{X}_{1:N}) \bigg| \overline{X}_{1:n} = x_{1:n} }$. Let
\begin{eqnarray}
Z_{2s-1}^{(\particleindex)} & = & \left [ \widetilde{f}_{\tau_{s}} \left ( X_{1:\tau_{s}}^{(\particleindex)} \right ) - \widetilde{f}_{\tau_{s-1}} \left ( \overline{X}_{1:\tau_{s-1}}^{(\particleindex)} \right ) \right ] H_{\tau_{s-1}}^{(\particleindex)}, \nonumber \\
Z_{2s}^{(\particleindex)} & = & \widetilde{f}_{\tau_{s}} \left ( \overline{X}_{1:\tau_{s}}^{(\particleindex)} \right ) H_{\tau_{s}}^{(\particleindex)} - \sum_{j=1}^{\nparticles} V_{\tau_{s}}^{(j)} \widetilde{f}_{\tau_{s}} \left ( X_{1:\tau_{s}}^{(j)} \right ) \widetilde{H}_{\tau_{s}}^{(j)}. \nonumber 
\end{eqnarray}

\subsection{The Main Result}

\begin{prop}\label{prop:unbiased}
$ \left \{ \left ( Z_{k}^{(1)}, \ldots, Z_{k}^{(M)} \right ), \FF_{k} : 1 \leq k \leq 2r + 1 \right \} $ is a martingale difference series and
$$ M \left ( \widetilde{\psi}_{OR} - \psi_{N} \right ) = \sum_{k=1}^{2r+1} \left ( Z_{k}^{(1)} + \cdots + Z_{k}^{(\nparticles)} \right ) $$
\end{prop}

\begin{rem}
It follows from the proposition that 
\begin{eqnarray}
\E{ \widetilde{\psi}_{OR} } & = & \psi_{N} \nonumber \\
\Rightarrow \E{ \frac{ Z_{N}^{\nparticles} }{ Z_{N} } \widehat{\psi}_{OR} } & = & \psi_{N} \nonumber \\
\Rightarrow \E{  Z_{N}^{\nparticles} \widehat{\psi}_{OR} } & = & Z_{N} \psi_{N}. \nonumber
\end{eqnarray}
Taking $\psi \equiv 1$ yields the desired result.
\end{rem}

\begin{proof}[Proof of Proposition \ref{prop:unbiased}]
We observe that 
\begin{eqnarray}
Z_{1}^{(\particleindex)} + \cdots + Z_{2r+1}^{(\particleindex)} & = & \sum_{s=1}^{r} Z_{2s}^{(\particleindex)} + \sum_{s=1}^{r+1} Z_{2s-1}^{(\particleindex)} \nonumber \\
& = & \sum_{s=1}^{r} \left [ \widetilde{f}_{\tau_{s}} \left ( \overline{X}_{1:\tau_{s}}^{(\particleindex)} \right ) H_{\tau_{s}}^{(\particleindex)} - \sum_{j=1}^{\nparticles} V_{\tau_{s}}^{(j)} \widetilde{f}_{\tau_{s}} \left ( X_{1:\tau_{s}}^{(j)} \right ) \widetilde{H}_{\tau_{s}}^{(j)} \right ] \nonumber \\ 
& & + \hspace{0.1in} \sum_{s=1}^{r+1} \left [ \widetilde{f}_{\tau_{s}} \left ( X_{1:\tau_{s}}^{(\particleindex)} \right ) - \widetilde{f}_{\tau_{s-1}} \left ( \overline{X}_{1:\tau_{s-1}}^{(\particleindex)} \right ) \right ] H_{\tau_{s-1}}^{(\particleindex)} \nonumber \\
& = & - \widetilde{f}_{\tau_{0}} \left ( \overline{X}_{\tau_{0}}^{(\particleindex)} \right ) H_{\tau_{0}}^{(\particleindex)} + \sum_{s=1}^{r+1}  \widetilde{f}_{\tau_{s}} \left ( X_{1:\tau_{s}}^{(\particleindex)} \right ) H_{\tau_{s-1}}^{(\particleindex)} - \sum_{s=1}^{r} \sum_{j=1}^{\nparticles} V_{\tau_{s}}^{(j)} \widetilde{f}_{\tau_{s}} \left ( X_{1:\tau_{s}}^{(j)} \right ) \widetilde{H}_{\tau_{s}}^{(j)}  \nonumber
\end{eqnarray}
Thus,
$$
\sum_{\particleindex=1}^{\nparticles} \left ( Z_{1}^{(\particleindex)} + \cdots + Z_{2r+1}^{(\particleindex)} \right ) 
$$
$$
=  - \nparticles \widetilde{f}_{\tau_{0}} \left ( \overline{X}_{\tau_{0}}^{(\particleindex)} \right ) H_{\tau_{0}}^{(\particleindex)} + \sum_{\particleindex=1}^{\nparticles} \sum_{s=1}^{r+1}  \widetilde{f}_{\tau_{s}} \left ( X_{1:\tau_{s}}^{(\particleindex)} \right ) H_{\tau_{s-1}}^{(\particleindex)} - \sum_{\particleindex=1}^{\nparticles} \sum_{s=1}^{r} \sum_{j=1}^{\nparticles} V_{\tau_{s}}^{(j)} \widetilde{H}_{\tau_{s}}^{(j)}  \widetilde{f}_{\tau_{s}} \left ( X_{1:\tau_{s}}^{(j)} \right )  
$$
$$
 =  - M \widetilde{f}_{\tau_{0}} \left ( \overline{X}_{\tau_{0}}^{(\particleindex)} \right ) H_{\tau_{0}}^{(\particleindex)} + \sum_{s=1}^{r+1} \sum_{\particleindex=1}^{\nparticles} \widetilde{f}_{\tau_{s}} \left ( X_{1:\tau_{s}}^{(\particleindex)} \right ) H_{\tau_{s-1}}^{(\particleindex)} - \sum_{\particleindex=1}^{\nparticles} \sum_{s=1}^{r} \sum_{j=1}^{\nparticles} \frac{1}{\nparticles} H_{\tau_{s}-1}^{(j)} \widetilde{f}_{\tau_{s}} \left ( X_{1:\tau_{s}}^{(j)} \right ) 
$$
$$
 =  \sum_{\particleindex=1}^{\nparticles} \widetilde{f}_{\tau_{r+1}} \left ( \overline{X}^{(\particleindex)}_{\tau_{r+1}} \right ) H_{\tau_{r}}^{(\particleindex)} - \nparticles \widetilde{f}_{\tau_{0}} \left ( \overline{X}_{\tau_{0}}^{(\particleindex)} \right ) H_{\tau_{0}}^{(\particleindex)}
$$
$$
=  \nparticles \left ( \widetilde{\psi}_{OR} - \psi_{N} \right ) 
$$

The second inequality is from \ref{eq:3.new}, and the last equality is because $\tau_{0} = 0$ and
\begin{eqnarray}
\widetilde{f}_{\tau_{r+1}} \left ( x_{1:\tau_{r+1}} \right ) & = & \EE_{q} \sqBK{ \psi \left ( \overline{X}_{1:N} \right ) L_{N} \left ( \overline{X}_{1:N} \right ) | \overline{X}_{1:\tau_{r+1}} = x_{1:\tau_{r+1}} } = \psi \left ( x_{1:\tau_{r+1}} \right ) L_{N} \left ( x_{1:\tau_{r+1}} \right )  \nonumber \\
\Rightarrow \widetilde{f}_{\tau_{r+1}} \left ( X_{1:\tau_{r+1}}^{(\particleindex)} \right ) H_{\tau_{r}}^{(\particleindex)} & = & L_{N} \left ( X_{1:N}^{(\particleindex)} \right ) \psi \left ( X_{1:N}^{(\particleindex)} \right ) H^{(\particleindex)}_{\widetilde{\tau}(N)} \hspace{0.2in} \textrm{as } X_{1:\tau_{r+1}}^{(\particleindex)} = X_{1:N}^{(\particleindex)} \textrm{ and } \widetilde{\tau}(N) = \tau_{r}. \nonumber
\end{eqnarray}

Let $\EE_{\nparticles}$ denote expectation under the $\nparticles$ particle system. To prove the martingale difference property, we observe that

\begin{eqnarray}
\CCE{ \nparticles }{ Z_{2}^{(\particleindex)} }{ \FF_{1} } 
& = &
 \CCE{\nparticles}{ \widetilde{f}_{\tau_{1}} \left ( \overline{X}_{1:\tau_{1}}^{(\particleindex)} \right ) H_{\tau_{1}}^{(\particleindex)} - \sum_{j=1}^{\nparticles} V_{\tau_{1}}^{(j)} \widetilde{f}_{\tau_{1}} \left ( X_{1:\tau_{1}}^{(j)} \right ) \widetilde{H}_{\tau_{1}}^{(j)} }{ \FF_{1} } \nonumber \\
& = & \CCE{\nparticles}{  \widetilde{f}_{\tau_{1}} \left ( \overline{X}_{1:\tau_{1}}^{(\particleindex)} \right ) H_{\tau_{1}}^{(\particleindex)} - \frac{1}{\nparticles} \sum_{j=1}^{\nparticles} \widetilde{f}_{\tau_{1}} \left ( X_{\tau_{1}}^{(j)} \right ) H_{\tau_{0}}^{(j)} }{ \FF_{1} } \hspace{0.3in} \textrm{from } \eqref{eq:3.new} \nonumber \\
& = & 
\CCE{ \nparticles }{ \widetilde{f}_{\tau_{1}} \left ( \overline{X}_{1:\tau_{1}}^{(\particleindex)} \right ) H_{\tau_{1}}^{(\particleindex)} - \frac{1}{\nparticles} \sum_{j=1}^{\nparticles} \widetilde{f}_{\tau_{1}} \left ( X_{\tau_{1}}^{(j)} \right ) }{ \FF_{1} } 
\hspace{0.57in} \textrm{as } H_{0}^{(j)} =1 \nonumber \\
& = & 0. \nonumber
\end{eqnarray}
The last equality is because the conditional distribution of $ \left ( \overline{X}_{1:\tau_{s}}^{(1)}, \ldots, \overline{X}_{1:\tau_{s}}^{(\nparticles)} \right )$ given $\FF_{2s-1}$ is that of $\nparticles$ i.i.d. random vectors which take the value $X_{1:\tau_{s}}^{(j)}$ with probability $V_{\tau_{s}}^{(j)}$. 
Also,
\begin{eqnarray}
\CCE{\nparticles}{ Z_{3}^{(\particleindex)} }{ \FF_{2} } & = & \EE \left [ \left \{ \widetilde{f}_{\tau_{2}} \left ( X_{1:\tau_{2}}^{(\particleindex)} \right ) - \widetilde{f}_{\tau_{1}} \left ( \overline{X}_{1:\tau_{1}}^{(\particleindex)} \right ) \right \} H_{\tau_{1}}^{(\particleindex)} \bigg| { \FF_{2} } \right ] \nonumber \\
& = & \EE \left [ \left \{ \widetilde{f}_{\tau_{2}} \left ( X_{1:\tau_{2}}^{(\particleindex)} \right ) - \widetilde{f}_{\tau_{1}} \left ( \overline{X}_{1:\tau_{1}}^{(\particleindex)} \right ) \right \} \bigg| { \FF_{2} } \right ] H_{\tau_{1}}^{(\particleindex)} \nonumber \\
& = & 0. \nonumber
\end{eqnarray}
The last equality is because
\begin{eqnarray}
\CCE{\nparticles}{ \widetilde{f}_{\tau_{s}} \left ( X_{1:\tau_{s}}^{(\particleindex)} \right ) }{ \FF_{2(s-1)} } & = & 
\EE_{\nparticles} \left [ \EE_{q} \left ( \psi ( X_{1:N} ) L_{N}(X_{1:N} ) \bigg| X_{1:\tau_{s}} = X_{1:\tau_{s}}^{(\particleindex)} \right ) \bigg| \FF_{2(s-1)} \right ] \nonumber \\
& = & \EE_{q} \left ( \psi \left ( X_{1:N} \right) L_{N} \left ( X_{1:N} \right ) \bigg| X_{1:\tau_{s-1}} = X_{1:\tau_{s-1}}^{(\particleindex)} \right ) \nonumber \\
& = & \widetilde{f}_{\tau_{s-1}} \left ( \overline{X}_{1:\tau_{s-1}}^{(\particleindex)} \right ). \nonumber 
\end{eqnarray}
The last equality is by the tower property of conditional expectations. Proceeding in this way, it is seen that $ \left \{ \left ( Z_{k}^{(1)}, \ldots, Z_{k}^{(\nparticles)} \right ), \FF_{k} : 1 \leq k \leq 2r + 1 \right \} $ is a martingale difference sequence.
\end{proof}

\end{document}